\documentclass[11pt,letterpaper]{article}
\usepackage[letterpaper,margin=1in]{geometry}
\usepackage{fullpage}
\usepackage{times}
\usepackage{hyperref}
\hypersetup{
pdfborder = 0 0 0,
colorlinks=true,
citecolor=blue,
}
\usepackage{custom}          
\usepackage[numbers]{natbib} 
\usepackage{paralist}        
\usepackage{graphicx}        
\usepackage{algorithm}       
\usepackage[noend]{algpseudocode}

\usepackage{amsmath}
\usepackage{amsthm}

\newtheorem{theorem}{Theorem}[section]
\newtheorem{lemma}[theorem]{Lemma}
\newtheorem{corollary}[theorem]{Corollary}

\usepackage{color}
\newcounter{todo}%

\def\beep{{\kw beep} }

\begin{document}
  \title{Computing a Maximal Independent Set Using Beeps}
  \author{Alejandro Cornejo\footnote{\texttt{acornejo@mit.edu},
          Massachusetts Institute of Technology (MIT)}, Bernhard
      Haeupler\footnote{\texttt{haeupler@mit.edu}, Massachusetts
          Institute of Technology (MIT)}, and Fabian
      Kuhn\footnote{\texttt{fabian.kuhn@usi.ch}, University of Lugano
          (USI), Switzerland}}
  \date{}
  \maketitle

  \begin{abstract}
    We consider the problem of finding a maximal independent set (MIS) in
the discrete beeping model introduced in DISC 2010. At each time, a node
in the network can either beep (i.e., emit a signal) or be silent.
Silent nodes can only differentiate between no neighbor beeping, or at
least one neighbor beeping. This basic communication model relies only
on carrier-sensing.  Furthermore, we assume nothing about the underlying
communication graph and allow nodes to wake up (and crash) arbitrarily.

We show that if a polynomial upper bound on the size of the network $n$
is known, then with high probability every node becomes stable in
$O(\log^3 n)$ time after it is woken up. To contrast this, we establish
a polynomial lower bound when no a priori upper bound on the network
size is known. This holds even in the much stronger model of local
message broadcast with collision detection.

Finally, if we assume nodes have access to synchronized clocks or we
consider a somewhat restricted wake up, we can solve the MIS
problem in $O(\log^2 n)$ time without requiring an upper bound on the
size of the network, thereby achieving the same bit complexity as Luby's
MIS algorithm.

  \end{abstract}

  \vspace{8.5cm}
  
  \pagebreak
  \setcounter{page}{2}

  \section{Introduction}

This paper studies the problem of computing a \emph{maximal
    independent set} (MIS) in the discrete beeping wireless network
model of~\cite{beepcolor}.  A maximal independent set of a graph is a
subset $S$ of vertices, such that no two neighboring vertices belong
to $S$, and any vertex outside of $S$ has a neighbor inside
$S$. Computing an MIS of a network in a distributed way is a classical
problem that has been studied in various communication models. On the
one hand, the problem is fundamental as it prototypically models
symmetry breaking, a key task in many distributed computations. On the
other hand, the problem is of practical interest, as especially in
wireless networks, having an MIS provides a basic clustering that can
be used as a building block for, e.g., efficient broadcast, routing,
or scheduling.

The network is modelled as a graph and time progresses in discrete and
synchronized time slots. In each time slot, a node can either transmit
a ``jamming'' signal (i.e., a beep) or detect whether at least one
neighbor beeps. We believe that such a model is minimalistic enough to
be implementable in many real world scenarios. At the same time, the
model is simple enough to study and mathematically analyze distributed
algorithms. Further, it has been shown that such a minimal
communication model is strong enough to efficiently solve non-trivial
tasks \cite{beepcolor,infocom09,schneider10disc}.

The beeping model can be implemented using only carrier sensing where
nodes need only to differentiate between silence and the presence of
close-by activity on the wireless channel. Note that we do \emph{not}
assume that nodes can sense the carrier and send a beep
simultaneously, a node that is beeping is assumed to receive no
feedback. We believe that the model is also interesting from a
practical point of view since carrier sensing can typically be used to
communicate more energy efficiently and over larger distances than
sending regular messages.

Besides the basic communication properties describe above, we make
almost no additional assumptions. Nodes wakeup asynchronously
(controlled by an adversary), and sleeping nodes are \emph{not}
automaticaly woken up by incoming messages. Upon waking up, a node has
no knowledge about the communication network. In particular a node has
no a priori information about its neighbors or their state.  No
restrictions is placed on the structure of the underlying
communication graph (e.g., it need not be a unit disk graph or a
growth-bounded graph).

Our contributions are two-fold. First, we show that if nodes are not
endowed with any information about the underlying communication graph,
any (randomized) distributed algorithm to find an MIS requires at
least $\Omega(\sqrt{n/\log n})$ rounds. We remark that this lower
bound holds much more generally.  We prove the lower bound for the
significantly more powerful radio network model with arbitrary message
size and collision detection, and is therefore not an artifact of the
amount of information which can be communicated in each
round. Furthermore, this lower bound can be easily extended for other
problems which require symmetry breaking (such as e.g., a
coloring or a small dominating set).

Second, we study what upper bounds can be obtained by leveraging some
knowledge of the network. Aided only by a polynomial upper bound on
the size of the network, we present a simple, randomized distributed
algorithm that finds an MIS in $\bigO(\log^3 n)$ rounds with high
probability. We then show that the knowledge of an upper bound on $n$
can be replaced by synchronous clocks.
In this case, we describe how to find an MIS with
high probability in $\bigO(\log^2 n)$ rounds. Finally, we show that the
synchronous clocks assumption can be simulated if we allow the wake up
pattern tu be slightly restricted,
also achieving a running time of $\bigO(\log^2 n)$.  We highlight that
all the upper bounds presented in this paper compute an MIS eventually
and almost surely, and thus only their running time is randomized.
Moreover, in addition to being robust to nodes waking up, with no
changes the algorithms also support nodes leaving the network with
similar guarantees. 

\paragraph{Related Work:}
The computation of an MIS has been recognized and studied as a
fundamental distributed computing problem for a long time (e.g.,
\cite{alon86,awerbuch89,luby86,panconesi95}).  Perhaps the single most
influential MIS algorithm is the elegant randomized algorithm of
\cite{alon86,luby86}, generally known as Luby's algorithm, which has a
running time of $\bigO(\log n)$. This algorithm works in a standard
message passing model, where nodes can concurrently reliably send and
receive messages over all point-to-point links to their
neighbors. \cite{sirocco09} show how to improve the bit complexity of
Luby's algorithm to use only $O(\log n)$ bits per channel ($O(1)$ bits
per round).  For the case where the size of the largest independent
set in the neighborhood of each node is restricted to be a constant
(known as bounded independence or growth-bounded graphs),
\cite{podc08} presented an algorithm that computes an MIS in
$\bigO(\log^* n)$ rounds. This class of graphs includes unit disk
graphs and other geometric graphs that have been studied in the
context of wireless networks.

The first effort to design a distributed MIS algorithm for a wireless
communication model is by \cite{mass04}. They provided an algorithm
for the radio network model with a $\bigO(\log^9 n/\log\log n)$
running time. This was later improved~\cite{podc05} to $\bigO(\log^2
n)$. Both algorithms assume that the underlying graph is a unit disk
graph (the algorithms also work for somewhat more general classes of
geometric graphs). The two algorithms work in the standard radio
network model of \cite{bgi} in which nodes cannot distinguish between
silence and the collision of two or more messages. The use of carrier
sensing (a.k.a.\ collision detection) in wireless networks has e.g.\
been studied in \cite{chlebus00,ilcinkas10,schneider10disc}. As shown
in \cite{schneider10disc}, collision detection can be powerful and can
be used to improve the complexity of algorithms for various basic
problems. \cite{mobihoc08} show how to approximate a minimum
dominating set in a physical interference (SINR) model where in
addition to sending messages, nodes can perform carrier sensing. In
\cite{flurywattenhofer}, it is demonstrated how to use carrier sensing as an
elegant and efficient way for coordination in practice.

The present paper is not he first one that uses carrier sensing alone
for distributed wireless network algorithms. A similar model to the
beep model considered here was first studied in
\cite{degesys07,infocom09}. As used here, the model has been
introduced in \cite{beepcolor}, where it is shown how to efficiently
obtain a variant of graph coloring that can be used to schedule
non-overlapping message transmissions. Most related to this paper are
results from \cite{schneider10disc} and \cite{science11}. In
\cite{schneider10disc}, it is shown that by solely using carrier
sensing, an MIS can be computed in $O(\log n)$ time in growth-bounded graphs
(a.k.a.\ bounded independence graphs). Here, we drop that restriction
and study the MIS problem in the beeping model for general graphs. In
\cite{science11}, Afek et al.\ described an $O(\log^2n)$ algorithm for
a similar model motivated by a biological process in the
development of the nervous system of flies. In \cite{science11}, it
assumed that nodes can beep and listen to neighboring beeps at the
same time and that all nodes are woken up synchronously.


  \section{System Model and Preliminary Definitions}
\label{sec:model}

In this paper we adopt the discrete beeping model introduced in
\cite{beepcolor}.
To model the communication network we assume there is an underlying
undirected graph $G=(V,E)$, where $V$ is a set of $n=|V|$ vertices  and
$E$ is the set of edges 
We denote the set of
neighbors of node $u$ in $G$ by $N_G(u)=\set{v \mid \set{u,v} \in E}$.
For a node $u \in V$ we use $d_G(u)=|N_G(u)|$ to denote its degree
(number of neighbors) and we use $d_{\max}=\max_{u \in V} d_G(u)$ to
denote the maximum degree of $G$.

We consider a synchronous network model where an adversary can choose
when a node wakes up and when it crashes.
Specifically, each node in $G$ is occupied by a process and the system
progresses in synchronous 
rounds.
Initially all processes are sleeping, and a process starts participating
at the round when it is woken up, which is chosen by an adversary.
At any round the adversary can furthermore remove nodes by 
making them crash (or leave) permanently. 
We denote by $G_t \subseteq G$ the subgraph induced by the processes
which are participating at round $t$. Note that we 
described the model for an oblivious adversary that chooses a fixed 
$G$ without knowing the randomness used by the algorithm. 

Instead of communicating by exchanging messages, 
we consider a more primitive communication model that
relies entirely on carrier sensing. Specifically, in every
round a participating process can choose to either beep or listen.
In a round where a process decides to beep it receives no feedback.
If a process at node $v$ listens in round $t$ it can
only distinguish between silence (i.e., no process $u \in N_{G_t}(v)$ beeps
in round $t$) or the presence of one or more beeps (i.e., there exists a
process $u \in N_{G_t}(v)$ who beeps in round $t$).
Observe that a beep conveys less information than a conventional 1-bit
message, since in the latter its possible to distinguish between
no message, a message with a one, and a message with a zero.

Given an undirected graph $H$, a set of vertices $I \subseteq V(H)$ is
an independent set of $H$ if every edge $e \in E$ has at most one
endpoint in $I$. An independent set $I \subseteq V(H)$ is a maximal
independent set of $H$, if for all $v \in V(H)\setminus I$ the set
$I\cup\set{v}$ is not independent.
An event is said to occur with high probability, if it occurs with
probability at least $1-n^{-c}$ for any constant $c \ge 1$, where
$n=|V|$ is the number of nodes in the underlying communication graph. 
For a positive integer integer $k \in \mathbb{N}$ we use $[k]$ as short
hand notation for $\set{1,\ldots,k}$. In a slight abuse of this notation
we use $[0]$ to denote the empty set $\varnothing$ and for $a,b \in
\mathbb{N}$ and $a < b$ we use $[a,b]$ to denote the set
$\set{a,\ldots,b}$.

This paper describes several distributed algorithms that find a
maximal independent set in the beeping model.
In the algorithms described in this paper, nodes can be in one of three
possible states: \emph{inactive}, \emph{competing} and
\emph{MIS}. We say a node is \emph{stable} if it is in the
MIS and all its neighbors are inactive, or if it has a
stable neighbor in the MIS.
Observe that by definition, if all nodes are stable then every node is
either in the MIS or inactive, and the MIS nodes describe a maximal
independent set.
We will focus solely on algorithms in which eventually all nodes become
stable (i.e., with probability one), and once nodes become stable they
remain stable unless an MIS node crashes.
In other words, we only consider Las Vegas type algorithms which always
produce the correct output, but whose running time is a random variable.
Moreover, we will show that with high probability nodes become stable
quickly.

We say a (randomized) distributed algorithm solves the MIS problem in $T$
rounds, if in the case that no wake ups and crashes happen for $T$ rounds
all nodes become stable with high probability. We furthermore say an MIS
algorithm is \emph{fast-converging} if it also guarantees that any individual
node irrevocably decides to be inactive or in the MIS after
being awake for at most $T$ rounds. Note that this stronger termination guarantee
makes only sense if there are no crashes since a stable inactive node has to change
its status and join the MIS if all its MIS neighbors crash. Moreover,
this is precisely the guarantee that we provide in the algorithm presented in Section
4.


\section{Lower Bound for Uniform Algorithms}
\label{sec:lowerbound}

In this section we show that without some a priori information about the
network (e.g., an upper bound on its size or maximum
degree) any fast-converging (randomized) distributed algorithm
needs at least polynomial time to find an MIS with constant probability.
In some ways, this result is the analog of the polynomial lower
bound~\cite{isaac02} on the number of rounds required for a successful
transmission in the radio network model without collision detection or
knowledge of $n$.

We stress that this lower bound is not an artifact of the beep model,
but a limitation that stems from having message
transmission with collisions and the fact that nodes are required to decide (but not
necessarily terminate) without waiting until all nodes have woken up.
Although we prove the lower bound for the problem of finding an
MIS, this lower bound can be generalized to other problems (e.g., minimal
dominating set, coloring, etc.).

Specifically, we prove the lower bound for the stronger communication
model of the local message broadcast with collision detection.
In this communication model a process can choose in every round either
to listen or to broadcast a message (no restrictions are made on the
size of the message). When listening a process receives silence if
no message is broadcast by its neighbors, it receives a collision
if a message is broadcast by two or more neighbors, and it receives
a message if it is broadcast by exactly one of its neighbors.
The beep communication model can be easily simulated by this model
(instead of beeping send a $1$ bit message, and when listening translate a
collision or the reception of a message to hearing a beep) and hence the
lower bound applies to the beeping model.

At its core, our lower bound argument relies on the observation that 
a node can learn essentially no information about the graph $G$ if upon waking up,
it always hears collisions or silence. It thus has to decide whether it
remains silent or beeps within a \emph{constant} number of rounds. More
formally:

\begin{prop} \label{prop:beep}
  Let $A$ be an algorithm run by all nodes, and let $b \in
  \{\mathrm{silent},\mathrm{collision}\}^*$ be a fixed pattern.
  If after waking up a node $u$ hears $b(r)$ whenever it listens in
  round $r$, then there are two constants $\ell \geq 1$ and $p \in
  (0,1]$ that only depend on $A$ and $b$ such that either
  \begin{inparaenum}[\bf a)]
  \item $u$ remains listening indefinitely, or 
  \item $u$ listens for $\ell-1$ rounds and broadcasts in round $\ell$
  with probability $p$.
  \end{inparaenum}
\end{prop}

\begin{proof}
  Fix a node $u$ and let $p(r)$ be the probability with which node $u$
  beeps in round $r$.  Observe that $p(r)$ can only depend on $r$, what
  node $u$ heard up to round $r$ (i.e., $b$) and its random choices.
  Therefore, given any algorithm, either $p(r)=0$ for all $r$ (and node
  $u$ remains silent forever), or $p(r) > 0$ for some $r$, in which case
  we let $p=p(r)$ and $\ell=r$.
\end{proof}

We now prove the main result of this section:

\begin{thm}
  \label{thm:lowerbound}
  If nodes have no a priori information about the graph $G$ then any
  fast-converging distributed algorithm in the local message broadcast
  model with collision detection that solves the MIS problem with
  constant probability requires requires at least $\Omega(\sqrt{n/\log
  n})$ rounds, even if no node crashes.
\end{thm}

\begin{proof}
Fix any algorithm $A$. Using the previous proposition we split the
analysis in three cases, and in all cases we show that with probability
$1-o(1)$  any algorithm runs for $o(\sqrt{n/\log n})$ rounds.
  
We first ask what happens with nodes running algorithm $A$ that hear
only silence after waking up.
Proposition \ref{prop:beep} implies that either nodes remain silent
forever, or there are constants $\ell$ and $p$ such that nodes broadcast
after $\ell$ rounds with probability $p$.
In the first case, suppose nodes are in a clique, and observe that no
node will ever broadcast anything. From this it follows that nodes cannot
learn anything about the underlying graph (or even tell if they are
alone). Thus, either no one joins the MIS, or all nodes join the MIS
with constant probability, in which case their success probability is
exponentially small in $n$.

Thus, for the rest of the argument we assume that nodes running $A$ that hear
only silence after waking up broadcast after $\ell$ rounds with
probability $p$.
Now we consider what happens with nodes running $A$ that hear
only collisions after waking up. Again, by
Proposition \ref{prop:beep} we know that either they remain silent
forever, or there are constants $m$ and $p'$ such that nodes 
broadcast after $m$ rounds with probability $p'$.
In the rest of the proof we describe a different execution for each of
these cases.


\paragraph{CASE 1: (a node that hears only collisions remains silent
forever)} (see \cref{fig:case1} in \cref{sec:figures})

    For some $k \gg \ell$ to be fixed later, we consider a set of $k-1$
    cliques $C_1,\ldots,C_{k-1}$ and a set of $k$ cliques
    $U_1,\ldots,U_{k}$, where each clique $C_i$ has $\Theta(k\log
    n/p)$ vertices, and each clique $U_j$ has $\Theta(\log n)$ vertices.
    We consider a partition of each clique $C_i$ into $k$ sub-cliques
    $C_i(1),\ldots,C_i(k)$ each with $\Theta(\log n/p)$ vertices.
    For simplicity, whenever we say two cliques are connected, 
    they are connected by a complete bipartite graph.

    Consider the execution where in round $i \in [k-1]$
    clique $C_i$ wakes up, and in round $\ell$ the
    cliques $U_1,\ldots,U_k$ wake up simultaneously.
    When clique $U_j$ wakes up, it is is connected to sub-clique $C_i(j)$
    for each $i < \ell$.
    Similarly, when clique $C_i$ wakes up, if $i \ge \ell$ then for $j
    \in [k]$ sub-clique $C_i(j)$ is connected to clique $U_j$.

    During the first $\ell-1$ rounds only the nodes in $C_1$ are
    participating, and hence every node in $C_1$ broadcasts in round
    $\ell+1$ with probability $p$. Thus w.h.p.\
    for all $j \in [k]$ at least two nodes in sub-clique $C_1(j)$
    broadcast in round $\ell$.  This guarantees that all nodes in 
    cliques $U_1,\ldots,U_k$ hear a collision during the first round
    they are awake, and hence they also listen for the second
    round.  In turn, this implies that the nodes in $C_2$ hear
    silence during the first $\ell-1$ rounds they participate, and again
    for $j \in [k]$ w.h.p.\ there  are at least two
    nodes in $C_2(j)$ that broadcast in round $\ell+2$.

    By a straightforward inductive argument we can show (omitted) that in general
    w.h.p.\ for each $i \in [k-1]$ and for every $j \in
    [k]$ at least two nodes in sub-clique $C_i(j)$ broadcast in
    round $\ell+i$.
    Therefore, also w.h.p., all nodes in cliques
    $U_1,\ldots,U_k$ hear collisions during the first $k-1$ rounds
    after waking up.

    Observe that at most one node in each $C_i$ can join the MIS (i.e.
    at most one of the sub-cliques of $C_i$ has a node in the MIS),
    which implies there exists at least one clique $U_j$ that is
    connected to only non-MIS sub-cliques. However, since the nodes in
    $U_j$ are connected in a clique, exactly one node of $U_j$ must
    decide to join the MIS, but all the nodes in $U_j$ have the same
    state during the first $k-1$ rounds. Therefore if nodes decide after
    participating for at most $k-1$ rounds, w.h.p.\ either no one in $U_j$ joins the MIS, or more than
    two nodes join the MIS.

    Finally since we have $n \in \Theta(k^2\log n+k\log n)$ nodes, we
    can let $k \in \Theta(\sqrt{n/\log n})$ and the theorem
    follows. 



\paragraph{CASE 2: (a node that hears only collisions remains silent forever)}
(see \cref{fig:case2} in \cref{sec:figures})

    For some $k \gg m$ to be fixed later let $q=\floor{\frac{k}{4}}$ and
    consider a set of $k$ cliques $U_1,\ldots, U_k$ and a set of $m-1$
    cliques $S_1,\ldots,S_{m-1}$, where each clique $U_i$ has
    $\Theta(\log n/p')$ vertices, and each clique $S_i$ has
    $\Theta(\log n/p)$ vertices.
    As before, we say two cliques are connected if they form a complete
    bipartite graph.

    Consider the execution where in round $i \in [m-1]$
    clique $S_i$ wakes up, and in round $\ell+j$ for $j \in [k]$
    clique $U_j$ wakes up.
    When clique $U_j$ wakes up, if $j > 1$ it is connected to every
    $U_i$ for $i \in \set{\max(1,j-q),\ldots,j-1}$ and if $j < m$ it is
    also connected to every clique $S_h$ for $h \in \set{m-j,\ldots,m}$.

    During the first $\ell-1$ rounds only the nodes in $S_1$ are
    participating, and hence every node in $S_1$ broadcasts in round
    $\ell+1$ with probability $p$, and thus w.h.p.\ at least two nodes in $S_1$ broadcast in round $\ell+1$.
    This guarantees the nodes in $U_1$ hear a collision upon waking up, and
    therefore they listen in round $\ell+2$.
    In turn this implies the nodes in $S_2$ hear silence during the
    first $\ell-1$ rounds they participate, and hence w.h.p.\ at least two nodes in $S_2$ broadcast in round $\ell+2$.

    By a straightforward inductive argument we can show (omitted) that in general
    for $i \in [m-1]$ the nodes in $S_i$ hear silence for the first
    $\ell-1$ rounds they participate, and w.h.p.\ at least
    two nodes in $S_i$ broadcast in round $\ell+i$.
    Moreover, for $j
    \in [k]$ the nodes in $U_j$ hear collisions for the first $m-1$ rounds
    they participate, and hence w.h.p.\ there are at least
    two nodes in $U_j$ who broadcast in round $\ell+m+j-1$. This implies
    that w.h.p.\  for $j \in [k-q]$ the nodes in $U_j$ hear
    collisions for the first $q$ rounds they participate.

    We argue that if nodes choose weather or not to join the MIS $q$
    rounds after participating, then they fail w.h.p.\ In particular
    consider the nodes in clique $U_j$ for $j \in
    \set{q,\ldots,k-2q}$. These nodes will collisions for the first
    $q$ rounds they participate, and they are connected to other nodes
    which also hear beeps for the first $q$ rounds they participate.
    Therefore, if nodes decide after participating for less or equal
    than $q$ rounds, w.h.p.\ either a node and all its neighbors won't
    be in the MIS, or two or more neighboring nodes join the MIS.

    Finally since we have $n \in \Theta(m\log n + k\log n)$ nodes, we can
    let $k \in \Theta(n/\log n)$ and hence $q \in \Theta(n/\log n)$ and
    the theorem follows.
\end{proof}


  \section{Maximal Independent Sets Using an Upper Bound on \boldmath$n$} \label{sec:alg1}

\def\plength{c \log N}

In this section we describe a simple and robust randomized
distributed algorithm that computes an MIS with high probability in a
polylogarithmic number of rounds.
Specifically, the algorithm only
requires an upper bound $N > n$ on the total number of nodes in the system, and
guarantees that with high probability, $\bigO(\log^2 N \log n)$ rounds
after joining, a node knows if it belongs to the MIS or if it is covered
by an MIS node.
Therefore, if the known upper bound is
polynomial in $n$ (i.e., $N \in \bigO(n^c)$  for a constant $c$), the
algorithm terminates with high probability in time $\bigO(\log^3 n)$.

\paragraph{Algorithm:}

If a node hears a beep while listening at any point during the
execution, it restarts the algorithm.
When a node wakes up (or it restarts), it stays in an inactive state
where it listens for $c\log^2 N$ consecutive rounds.
After this inactivity period, nodes start competing and group rounds
into $\log N$ phases of $\plength$ consecutive rounds. Due to the
asynchronous wake up and the restarts, in general phases of different
nodes will not be synchronized.
In each round of phase $i$ with probability $2^i/8 N$ a node beeps, and
otherwise it listens. Thus by phase $\log N$ a node beeps
with probability $\frac{1}{8}$ in every round.
After successfully going through the $\log N$ phases of the competition
(recall that when a beep is heard during any phase, the algorithm
restarts) a node assumes it has joined the MIS and into a loop where it
beeps in every round with probability $1/8$ forever (or until it hears a
beep).

\begin{algorithm}
  \caption{FastMIS algorithm}
  \begin{algorithmic}[1]
    \State {\bf for} $c \log^2 N$ rounds {\bf do} listen \Comment{Inactive}
    \For{$i \in \set{1,\ldots,\log N}$} \Comment{Competing}
      \For{$\plength$ rounds}
      \State with probability $2^i/8 N$ beep, otherwise listen
      \EndFor
    \EndFor
    \State {\bf forever} with probability $\half$ beep then listen,
    otherwise listen then beep
    \Comment{MIS}
  \end{algorithmic}
\end{algorithm}

In contrast to the polynomial lower bound from Section \ref{sec:lowerbound} we 
show that the above algorithm does not only solve the MIS problem in $O(\log^2 N \log n)$
time but is also fast-converging. 

\begin{theorem}\label{thm:simpleMIS}
    The FastMIS algorithm solves the MIS problem in $O(\log^2 N \log
    n)$ time, where $N$ is an upper bound for $n$ that is a priori
    known to the nodes. Under arbitrary wake ups and no crashes, the
    FastMIS algorithm is furthermore fast-converging.
\end{theorem}

This demonstrates that knowing a priori size information about the network, even as 
simple as its size, can drastically change the complexity of a
problem. The knowledge of $n$ alone provablg creates an exponential
for the
running time of fast-converging MIS algorithms. 

\paragraph{Proof Outline.} First, we leverage the fact that for two neighboring
nodes to go simultaneously into the MIS they have to choose the same
actions (beep or listen) during at least $c\log N$ rounds. This does
no happen w.h.p.\ and thus MIS nodes are independent w.h.p.
On the other hand, since nodes which are in the MIS keep trying to break
ties, an inactive node will never become active while it has a neighbor
in the MIS, and even in the low probability event that two neighboring
nodes do join the MIS, one of them will eventually and almost surely
leave the MIS. 
The more elaborate part of the proof is showing that w.h.p., any node becomes stable after $O(\log^2 N\log n)$
consecutive rounds without crashes.  This requires three technical
lemmas. First we show that if the sum of the beep probabilities of
a neighbor are greater than a large enough constant, then they have
been 
than a (smaller) constant for the $c\log N$ preceding rounds. This can
be used this to
show that with constant probability, when a node $u$ hears or produces
beep, no neighbor of the beeping node beeps at the same time and thus
$u$ becomes
stable. Finally, since a node hears a beep or produces a beep every
$O(\log^2 N)$ rounds, $\bigO(\log^2 N \log n)$ rounds suffice to
stabilize w.h.p. (Detailed proofs in Appendix \ref{app:proofsalg1}.)


  \section{Synchronized Clocks}\label{sec:synch}

For this section we assume that nodes have synchronized clocks, i.e.,
know the current round number $t$. As before, we allow arbitrary node
additions and deletions.

\paragraph{Algorithm:}
Nodes have three different internal states:
\emph{inactive}, \emph{competing}, and \emph{MIS}.  Each node has a
parameter $k$ that is monotone increasing during the execution of the
algorithm.  All nodes start in the inactive state with $k=6$.

Nodes communicate in beep-triples, and synchronize by starting a
triple only when $t \equiv 0 \pmod 3$. The first bit of the triple is
the Restart-Bit. A \beep is sent for the Restart-Bit if and only if $t
\equiv 0 \pmod k$. If a node hears a \beep on its Restart-Bit it
doubles its $k$ and if it is active it becomes inactive. The second
bit sent in the triple is the MIS-Bit. A \beep is sent for the MIS-Bit
if and only if a node is in the MIS state. If a node hears a \beep on
the MIS-bit it becomes inactive. The last bit send in a triple is the
Competing-Bit. If inactive, a node listens to this bit, otherwise it
sends a \beep with with probability 1/2. If a node hears a \beep on
the Competing-Bit it becomes inactive. Furthermore, if a node is in
the MIS-state and hears a \beep on the Competing-Bit it doubles its
$k$. Lastly, a node transitions from inactive to active between any
time $t$ and $t+1$ for $t\equiv 0\pmod k$. Similarly, if a node is active when $t
= 0 \mod k$ then it transitions to the MIS state. In the sequel, we
refer to this algorithm as Algorithm 2. The state transitions 
are also depicted in Figure \ref{fig:synch}.

\begin{figure}[ht]
 \centering
  \includegraphics[height=2.5cm]{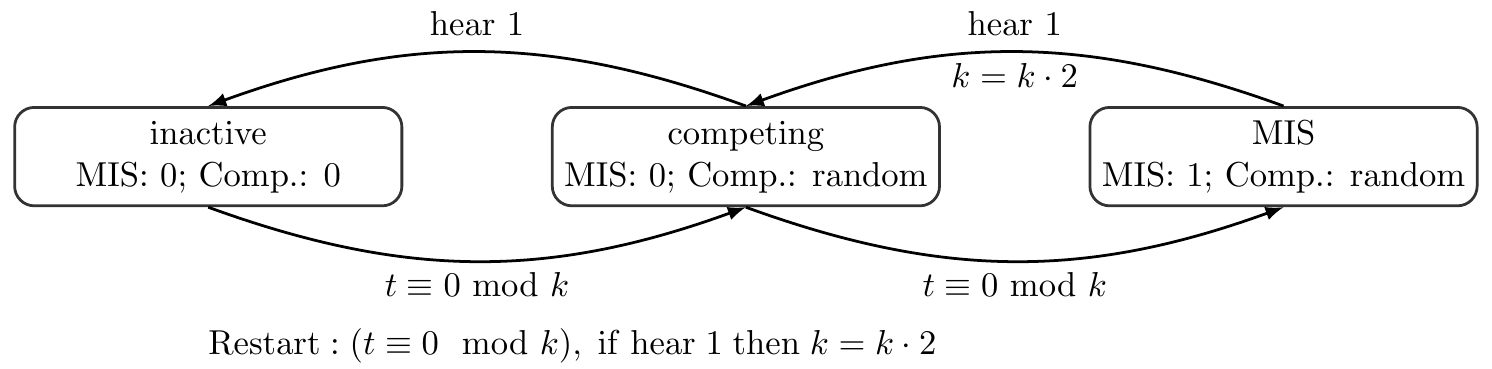}
  \caption{State Diagram for Algorithm 2}
  \label{fig:synch}
\end{figure}

\paragraph{Idea:}

The idea of the algorithm is to employ Luby's permutation algorithm in
which a node picks a random $O(\log n)$-size priority which it shares
with its neighbors. A node then joins the MIS if it has
the highest priority among its neighbors, and all neighbors of an MIS
node become inactive.
Despite the fact that this algorithm is described for
the message exchange model, it is straightforward to adapt the priority comparisons
to the \beep model. For this, a node sends its priority bit by
bit, starting with the highest-order bit and using a \beep for a $1$. The
only further modification is that a node stops sending its priority if
it has already heard a \beep on a higher order bit during which it
remained silent because it had a zero in the corresponding bit. Using this
simple procedure, a node can easily realize when a neighboring node has
a higher priority.
Furthermore, a node can observe that it has the highest-priority in its
neighborhood which is exactly the case if it does not hear any \beep. 

Therefore, as long as nodes have a synchronous start and know $n$ (or an
upper bound) it is straightforward to get Luby's algorithm working in the
beep model in $O(\log^2 n)$ rounds
(and ignoring edge additions and deletions). We
remark that this already implies a better round complexity that the
result of \cite{science11} in a strictly weaker model, albeit without
using a biologically inspired algorithm.

In the rest of this section we show how to remove the need for an upper
bound on $n$ and a synchronous start.  We solely rely on synchronized
clocks to synchronize among nodes when a round to transmits a new
priority starts. Our algorithm uses $k$ to compute an estimate for the
required priority-size $O(\log n)$. Whenever a collision occurs and two
nodes tie for the highest priority the algorithm concludes that $k$ is
not large enough yet and doubles its guess.
The algorithm furthermore uses the Restart-Bit to ensure that nodes
locally work with the same $k$ and run in a synchronized manner in which
priority comparisons start at the same time (namely every $t \equiv 0 \pmod
k$). It is not obvious that either a similar $k$ or a synchronized
priority comparison is necessary but it turns out that algorithms
without them can stall for a long time. In the first case this is
because repeatedly nodes with a too small $k$ enter the MIS state
simultaneously while in the second case many asynchronously competing
nodes (even with the same, large enough $k$) keep eliminating each other
without one becoming dominant and transitioning into the MIS state.

\paragraph{Analysis:}

To proof the algorithm's correctness, we first show two lemmas that show
that with high probability $k$ cannot be super-logarithmic.

\begin{lemma}\label{lem:smallk}
With high probability $k\in O(\log n)$ for all nodes during the execution of the algorithm.
\end{lemma}
\begin{proof}
We start by showing that two neighboring nodes $u,v$ in the MIS state
must have the same $k$ and transitioned to the MIS state at the same
time. We prove both statements by contradiction. 

For the first part assume
that nodes $u$ and $v$ are in the MIS state but $u$ transitioned to this
state (the last time) before $v$. In this case $v$ would have received
the MIS-bit from $u$ and become inactive instead of joining the MIS, a
contradiction.

Similarly, for sake of contradiction, we assume that $k_u < k_v$. In
this case, during the active phase of $u$ before it transitioned to the
MIS at time $t$ it would have set its Restart-bit to 0 at time $t - k_u$
and received a 1 from $v$ and become inactive, contradicting the
assumption that $k_u<k_v$.

Given this we now show that for a specific node $u$ it is unlikely to
become the first node with a too large $k$. For this we note that
$k_u$ gets doubled because of a Restart-Bit only if a \beep from a
node with a larger $k$ is received. This node can therefore not be
responsible for $u$ becoming the first node getting a too large
$k$. The second way $k$ can increase is if a node transitions out of
the MIS state because it receives a Competing-Bit from a neighbor
$v$. In this case, we know that $u$ competed against at least one such
neighbor for $k$ rounds with none of them loosing.  The probability of
this to happen is $2^{-k}$. Hence, if $k \in \Theta(\log n)$, this
does not happen w.h.p. A union bound over all nodes and the polynomial
number of rounds in which nodes are not yet stable finishes the proof.
\end{proof}

\begin{theorem}\label{thm:synchalg}
    If during an execution the $O(\log n)$ neighborhood of a node $u$
    has not changed for $\Omega(\log^2 n)$ rounds then $u$ is stable,
    i.e., $u$ is either in the MIS state with all its neighbors being
    inactive or it has at least one neighbor in the MIS state whose
    neighbors are all inactive.
\end{theorem}
\begin{proof}
First observe that if the whole graph has the same value of $k$ and no
two neighboring nodes transition to the MIS state at the same time, then
our algorithm behaves exactly as Luby's original permutation algorithm,
and therefore terminates after $O(k \log n)$ rounds with high
probability.
From a standard locality argument, it follows that a node $u$ also
becomes stable if the above assumptions only hold for a $O(k \log n)$
neighborhood around $u$. Moreover, since Luby's algorithm performs
only $O(\log n)$ rounds in the message passing model, we can improve our
locality argument to show that in if a $O(\log n)$ neighborhood around
$u$ is well-behaved, then $u$ behaves as in Luby's algorithm.

Since the values for $k$ are monotone increasing and propagate between
two neighboring nodes $u,v$ with different $k$ (i.e., $k_u > k_v$) in at
most $2k_u$ steps, it follows that for a node $u$
it takes at most $O(k_u \log n)$ rounds until either $k_u$
increases or all nodes $v$ in the $O(\log n)$ neighborhood of $u$ have
$k_v = k_u$ for at least $O(k \log n)$ rounds. We can furthermore assume
that these $O(k \log n)$ rounds are collision free (i.e, no two
neighboring nodes go into the MIS), since any collision leads with high
probability within $O(\log n)$ rounds to an increased $k$ value for one
of the nodes. 

For any value of $k$, within $O(k \log n)$ rounds a node thus either
performs Luby's algorithm for $O(\log n)$ priority exchanges, or it
increases its $k$. Since $k$ increases in powers of two and, according
to Lemma \ref{lem:smallk}, with high probability it does not
exceed $O(\log n)$, after at most $\sum_i^{O(\log \log n)} 2^i\cdot
3\cdot O(k \log n) = O(\log^2 n)$ rounds the status labeling around a
$O(\log n)$ neighborhood of $u$ is a proper MIS. This means that $u$ is
stable at some point and it is not hard to verify that the function of the
MIS-bit guarantees that this property is preserved for the rest of the
execution. 
\end{proof}

\section{Simple Wake Up}

In this section we show how to replace the assumption of
synchronized clocks by instead restricting the way in which
wake ups and crashes occur. The main theorem in this Section is Theorem
\ref{thm:independentsimplewakeup}.

We work with the following simple wake up restriction: The adversary is allowed to
start with any (possibly disconnected) graph, without loss of generality
we call this time $t=0$. Furthermore the adversary can at any time
wake up any set of new nodes, with the restriction that each new node is
connected at least to one \emph{old} node,
i.e., a node that has been around for at least
for $\delta$ rounds, where we think of $\delta$ as being a small 
non-constant quantity (e.g., $\log d_{\max}$). Similarly, the adversary
can crash any node, as long as this node is connected only to old nodes.

Given these 
quite flexible simple wake up dynamics, we will show that nodes
can simulate synchronous clocks. This reduction
allows us to execute the Algorithm 2 without synchronized clocks.
We start by presenting a very simple reduction, that requires $\delta$
to depend on the current round. 
We then refine the reduction
and show how to circumvent this problem and give an MIS
algorithm for the simple wake up assumption with $\delta$ that is at
least $\log d_{\max}$ (note that this does not imply that nodes need to
know $\log d_{\max}$).

\subsection{Simple Wake up and Synchronized Clocks}

The core idea is for each node to keep a local time counter, and use a
structured beep pattern to communicate this local time counter to new
nodes.
The simple wake up dynamics prevent the adversary from blocking these
beep pattern through staggered node additions, as those which were
described in the lower bound proof of Section \ref{sec:lowerbound}. 

We split messages into blocks.
A block starts with two zeros that unambiguously mark the beginning of a
block. This is followed by the current block $t$ (i.e. the time counter)
and lastly an equal amount of bits carrying the data of the simulated
algorithm. Both the individual bits describing the time and all data
bits are interleaved by ones which makes the block-beginning 
identifiable. As an example the bit sequence $abcdefghijkl$ would be
sent as
$00.0.a.00.1.b.00.1.0.c.d.00.1.1.e.f.00.1.0.0.g.h.i.00.1.0.1.j.k.l$,
where we replaced the separating ones by a period for better
readability. Observe that each block contains a header, the current
block (i.e., time), and some data.

The complete algorithm operates as follows. Once a node is awake, it
listens for four rounds. If no \beep was received during this time the
node can be sure that it is alone, which also implies that $t=0$. If a
node hears at least one \beep it waits until it hears two rounds of
silence in a row, which mark the beginning of a block. It then listens
for the length of the whole block which allows it to identify the
current block number $t$. In either case a node learns the current
block number (and thus the time) after listening for at most two
blocks. Then it is able to perform the same computations as the
synchronized algorithm of the previous section.

\begin{theorem}\label{thm:simluation}
Any algorithm that works in the discrete beep model with synchronized clocks in time $O(T)$ can be simulated by an algorithm that works for the simple wake up dynamics with $\delta = O(\log t)$ in time $O(T + \log t)$ where $t$ is the total time the algorithm is run. 
\end{theorem}
\begin{proof}
    If a new node is connected to a node that broadcasts the time for
    $O(\log t)$ rounds, it gets to know the number of blocks that have
    been sent (and thus time itself). In the simple wake up dynamics
    with $\delta = O(\log t)$ we thus get inductively that all old
    nodes are in synch and know the current time. With this knowledge
    they can easily infer how many data bits were sent around since
    the beginning of the algorithm and thus get a logical time on the
    data bits that is shared with all synchronized nodes. With this
    synchronization, old nodes can then run the simulated
    algorithm. Since during this computation by construction a
    constant fraction $c$ of the bits are data bits, the number of
    rounds the simulated algorithm needs to run after the $O(\log t)$
    time synchronization is at most $cT$.  This leads to the claimed
    total running time of $O(T + \log t)$.
\end{proof}

We can use the above reduction together with Algorithm 2 and obtain
an MIS algorithm with running time $O(\log^2 n + \log t)$. Unfortunately
if we run this algorithm for a long time we do not get any running time
guarantee in the number of nodes. To avoid this we make the observation
that Algorithm 2 does not use the full power of synchronized clocks but
solely requires that nodes can evaluate whether $t \equiv 0 \pmod k$ where
it suffices to have a $k = O(\log d_{\max})$ that is logarithmic in the
maximum degree of $G$. Thus if nodes know an a priori upper bound
$\Delta$ on $d_{\max}$ it suffices to track time modulo $\log \Delta$,
which requires only $\log \log \Delta$ time-bits. This way, at the cost
of having to know $\Delta$, the algorithm is not dependent on time any
more. We thus get the following corollary:

\begin{corollary}
    There is an algorithm that solves the MIS problem in a network
    with simple wake up dynamics with $\delta = O(\log t)$ in
    $O(\log^2 n + \log t)$ time, where $t$ is the time over which wake
    ups and crashes occur. If nodes are given an a priori upper bound
    $\Delta$ on the maximum degree then there is also a $O(\log^2 n)$
    time MIS algorithm that works in any network with simple wake up
    dynamics with $\delta = O(\log \log \Delta)$.
\end{corollary}

\subsection{The Simple Wake Up Algorithm}

In the last subsection we gave two algorithms that work in the simple
wake up model. Both algorithms have a drawback. The first one
deteriorates over time and thus the number of rounds it requires to
solve the MIS problem increases as time progresses.  On the other
hand, the second algorithm requires an a priori upper bound on
$d_{\max}$ or $n$. As we showed in Section \ref{sec:lowerbound} and
\ref{sec:alg1} this can be a drastic advantage for an algorithm. In
what follows we show that the synchronization required by Algorithm 2
can be achieved without a priori knowledge or dependence on $t$.

The algorithm builds on the synchronization ideas developed before where
nodes try to keep a time counter up to some precision. We know that it
suffices for the nodes to know the least significant $\log \log
d_{\max}$ bits. The problem is that now they do not know $d_{\max}$. The
following approach makes sure that nodes send out all time bits but
prioritize the earlier bits such that any node that listens for $O(2^l)$
rounds is able to infer the first $l$ bits of the time.

We will count the number of blocks that have been sent around since the
system started.  This allows us to start with $k=1$ and get all powers
of two for the values of $k$. It is important that nodes increase their
$k$.

Our approach is closely related to the binary carry sequence $B$, where the
$n$\ith digit is the number of zeros at the end of $n$ when written in
base 2.

$$B  =
0102010301020104010201030102010501020103010201040102010301020106010201\ldots$$
Suppose now we associate the $n^{th}$ number in this sequence with the
$n^{th}$ block that is sent. In this case if the binary carry sequence
number associated with a certain block is $l$, then for exactly all
nodes with $k = 2^i$ and $i \leq l$ we have $t \equiv 0 \mod k$ which
implies a status change and a zero in Restart-Bit for these nodes. We
are going to define another sequence $B'$ that allows to identify the
position of all numbers smaller than $l$ as long as any interval of
length $2^{l+2}$ of $B'$ is observed. The sequence $B'$ is a bit
sequence in which the $n^{th}$ bit is the parity of the number of
occurrences of the $n^{th}$ number of $B$ in the first $n$ numbers of
$B$, i.e.,
$$B' = 1101100111001001110110001100100111011001110010001101100011001001110110\ldots$$
Suppose that we split the sequence into two parts, one containing
all even $n$ and one with all odd $n$ (and thus corresponding to the
positions of zeros in $B$, and not zeros in $B$). Then the odd sequence is strictly alternating
while the even sequence is $1010$-free. The last observation is a result
of the fact that either the two ones or the two zeros correspond
to two consecutive occurrences of a $1$ in $B$ and should thus have
different parities. 
After having received (at most) eleven consecutive bits from sequence $B'$ a node will
therefore receive a $1010$ pattern on the odd subsequence and no such
pattern on the even subsequence. This allows it to identify the
positions of all zeros in $B$.
With this knowledge, a node can then turn its attention to the even
subsequence to identify the positions of the ones in $B$.
Again, this can be done in the same way,
namely splitting the even subsequence into two subsequences and waiting for a
$1010$ pattern. Iterating this procedure enables a node that receives $11 \cdot
2^l$ consecutive bits of the $B'$ sequence, to identify all positions of numbers
of at most $l$ in $B$. From then on, it can also send the right bits of
this sequence to its neighbors as soon as it learns them.

With this trick in mind we now describe the algorithm. It operates in
bit-quadruples in a similar manner to Algorithm 2 but with an
additional time-bit that is used to transmit the sequence $B'$ with the
time parity information. Furthermore use the block structure with the
leading zeros and the alternating ones to distinguish the beginning of a block.
In total a block with time-bit $T$, Restart-bit $R$, MIS-bit $M$ and
competing-bit $C$ will be transmitted as $00.T.R.M.C.$ where again the
$.$ represent the separating ones. The computation of the algorithm is
now as in Algorithm 2 with two minor modifications. First, the time bit
is used to convey the parity bit of the current time as described
above. A node constantly listens to bits in the sequence $B'$ that are
not known to it, learning more positions over time. On the position of
the bits it knows it sends these bits out. The second modification is
that a node never increases its $k$ value beyond the last power of two
for which it can evaluate $t \equiv 0 \pmod k$. This completes the
algorithm description. 

Given the arguments above we can prove our main theorem for this
section:

\begin{theorem}\label{thm:independentsimplewakeup}
There is a algorithm that solves the MIS problem in a network with
simple wake up dynamics with $\delta = O(\log d_{\max})$ in $O(\log^2
n)$ time.
\end{theorem}
\begin{proof}
We first observe by induction that if $\delta \in O(\log d_{\max})$ we get
that any old node knows about the $\log \log d_{\max}$ lowest
significant bits of time. This allows $k$ to grow as large as $O(\log
d_{\max})$ which is sufficient for Luby's algorithm to work
efficiently. Thus simulating running Algorithm 2 will be successful in
$O(\log^2 n)$ blocks (each of size $O(1)$). The only thing that we need
to show is
that two neighbors $u,v$ with different values of
$k_u > k_v$ still converge to the larger value $k_u$ in time $O(k_u)$
(if none of them crashes) even though $v$ might not be allowed to increase
its $k$ because it does not know the time well enough. This is true
because $v$ will learn the $k_u$ parity of the time in $O(k_u)$ rounds
and then increase its $k$ value accordingly. Besides this no further
modifications were made to the algorithm and its correctness and running
time therefore follow from Theorem \ref{thm:simluation} and
\ref{thm:synchalg}.
\end{proof}


  \pagebreak



  \newpage
  \noindent{\huge \bf Appendices}
  \bigskip
  \appendix
  
  \section{Figures}
\label{sec:figures}

\begin{figure}[h]
  \centering
  \includegraphics[width=\textwidth]{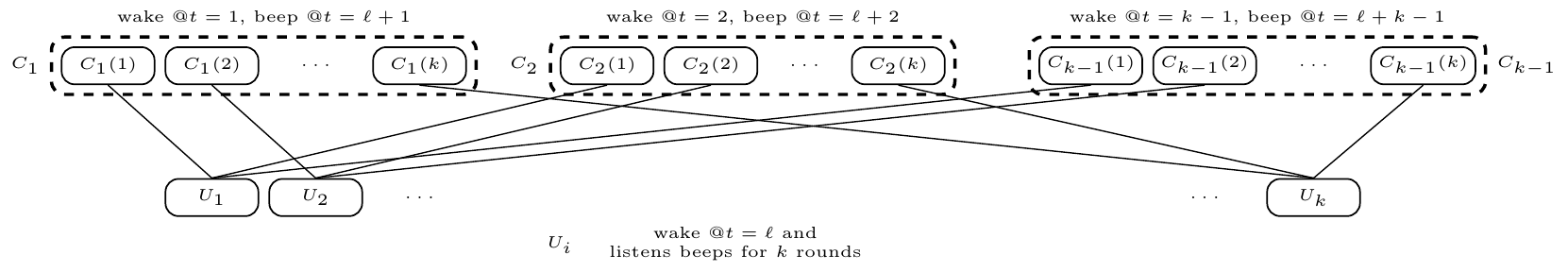}
  \caption{Execution for Case 1 of the Lower Bound}
  \label{fig:case1}
\end{figure}

\begin{figure}[h]
  \centering
  \includegraphics[width=0.6\textwidth]{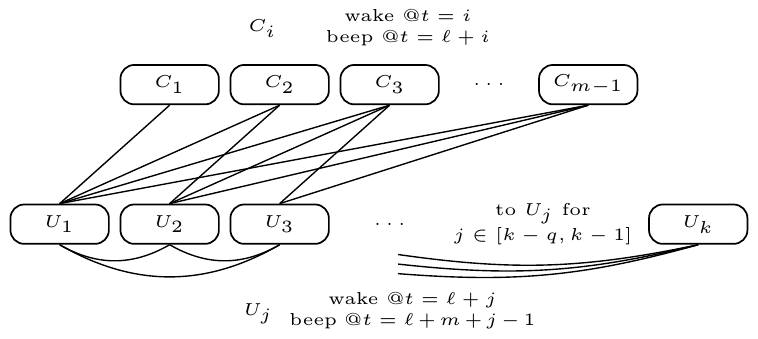}
  \caption{Execution for Case 2 of the Lower Bound}
  \label{fig:case2}
\end{figure}


  \section{Proofs for Section \ref{sec:alg1}}\label{app:proofsalg1}

First, we show that with high probability two neighboring nodes will
never join the MIS. Moreover, even in the low probability event that
they two neighboring nodes join the MIS, almost surely one of them
eventually becomes inactive.

\begin{claim}
  With high probability two neighboring nodes do not join the MIS. If
  two neighboring nodes are in the MIS, almost surely eventually one of
  them becomes inactive.
\end{claim}

\begin{proof}
  For two neighboring nodes to join the MIS they would first have to go
  through an interval of $c\log N$ consecutive rounds where at every
  round they both beep with probability $1/8$ and listen otherwise. 
  Moreover, during these $c\log N$ rounds it should not be the case
  that one of them listens while the other beeps, and hence they have to
  choose the same action (beep or listen) at each of these rounds.

  The probability of this happening is less than
  $(1-\frac{1}{8})^{c\log N} \le e^{-c\log N/8}$, and thus for
  sufficiently large $c$ (i.e $c \ge 8$) we have that with high
  probability two neighboring nodes do not join the MIS simultaneously.

  Moreover, assume two neighboring nodes are in the MIS simultaneously.
  Then at every round, one of them will leave the MIS with constant
  probability. Hence, the probability that they both remain in the MIS
  after $k$ rounds is exponentially small in $k$.  Hence, it follows
  that eventually almost surely one of them becomes inactive.
\end{proof}

Moreover, it is also easy to see that once a node becomes stable it
stays stable indefinitely (or until a neighbor crashes). This follows
by construction since stable MIS nodes will beep at least every 3
rounds, and therefore inactive neighbors will never start competing to
be in the MIS.
Hence to prove the correctness of the algorithm we need only
to show that eventually all nodes are either in the MIS or have a
neighbor in the MIS.

For a fixed node $u$ and a round $t$, we use $b_u(t)$ to denote the
\emph{beep probability} of node $u$ at round $t$.
The \emph{beep potential} of a set of nodes
$S \subseteq V$ at round $t$ is defined as the sum of the beep probabilities of
nodes in $S$ at round $t$, and denoted by 
$E_S(t)=\sum_{u \in S} b_u(t)$.
Of particular interest is the beep potential of the neighborhood of a
node, we will use $E_v(t)$ as short hand notation of
$E_{N(v)}(t)$.

The next lemma shows that if the beep potential of a particular set of
nodes is larger than a (sufficiently large) constant at round $t$, then
it was also larger than a constant during the interval $[t-\plength,t]$.
Informally, this is true because the beep probability of every node
increases slowly.

\begin{lem}
  \label{lem:slowchange}
  Fix a set $S \subseteq V$. If $E_S(t) \ge \lambda$ at round $t$,
  then $E_S(t') \ge \half\lambda-\frac{1}{8}$ at round $t' \in
  [t-\plength,t]$.
\end{lem}

\begin{proof}
  Let $P \subseteq S$ be the subset of nodes that are at phase
  $1$ at round $t$, and let $Q=S\setminus P$ be the
  remaining nodes. Using this partition of nodes we split the
  probability mass $E_S(t)$ as:

  \begin{equation}
  E_S(t)=\underbrace{\sum_{u \in P} b_u(t)}_{E_P(t)}+
  \underbrace{\sum_{u \in Q} b_u(t)}_{E_Q(t)}
  \end{equation}

  For the rest of the proof, let $t'$ be any round in the range $[t-\plength,t]$.
  Since the nodes in $P$ are in phase $1$ at round $t$, therefore at
  round $t'$ the nodes in $P$ are either in
  the inactive state or at phase $1$. This implies that $b_u(t')
  \le 1/4 N$ for $u \in P$, and since there are at most $|P| \le |S| \le
  N$ nodes we have $E_P(t') \le N/4 N = \frac{1}{4}$.

  Similarly the nodes in $Q$ are in phase $i > 1$ at round $t$, and
  therefore at round $t'$ the nodes in $Q$ are in
  phase $i-1\ge 1$. This implies that $b_u(t') \ge \half b_u(t)$ for $u
  \in Q$, and hence $E_Q(t') \ge \half E_Q(t) = \half(E_S(t)-E_P(t)) \ge
  \half\lambda-\frac{1}{8}$.

  Finally since $E_S(t') \ge E_Q(t')$ we have $E_S(t') \ge
  \half\lambda-\frac{1}{8}$.
\end{proof}

Using the previous lemma, we show that with high probability nodes which
are competing have neighborhoods with a ``low'' beep potential.
Intuitively this is true because if a node had neighborhoods with a ``high''
beep potential, by the previous result we know it also had a high beep
potential during the previous $c\log N$ rounds, and there are good
changes it would have been kicked out of the competition in a previous
round.

\begin{lem}
  \label{lem:lowEu}
  With high probability, if node $v$ is competing at round $t$ then $E_v(t) <
  \half$.
\end{lem}

\begin{proof}
  Fix a node $v$ and a time $t$, we will show that if $E_v(t)\ge \half$
  then with high probability node $v$ is not competing at time $t$.
  
  Let $L_v(\tau)$ be the event that node $v$ listens at round $\tau$ and
  there is a neighbor $u \in N(v)$ who beeps at round $\tau$. First we
  estimate the probability of the event $L_v(\tau)$.

  \begin{align*}
    \PR{L_v(\tau)} &= (1-b_v(\tau))\paren{1-\prod_{u \in
    N(v)}(1-b_u(\tau))} \ge
    (1-b_v(\tau))\paren{1-\exp\paren{-\sum_{u\in N(v)} b_u(\tau)}} \\
    &= (1-b_v(\tau))(1-\exp(-E_v(\tau)))
  \end{align*}

  From \cref{lem:slowchange} we have that if $E_v(t)\ge \half$ then
  $E_v(\tau) \ge \frac{1}{8}$ for $\tau \in [t-\plength, t]$, together
  with the fact that $b_v(\tau) \le \half$ this implies that $L_v(\tau)
  \ge \half(1-e^{-1/8}) \approx 0.058$ for $\tau \in [t-\plength,t]$.

  Let $C_v(t)$ be the event that node $v$ is competing at round $t$.
  Observe that if $L_v(\tau)$ occurs for $\tau \in [t-\plength,t]$
  then node $v$ stops competing for at least $\plength$ rounds and hence
  $C_v(t)$ cannot occur. Therefore, the probability that node $v$ does
  not beep at round $t$ is at least:

  \begin{align*}
    \PR{\neg C_v(t)} &\ge \PR{\exists \tau \in
    [t-\plength,t] \st L_v(\tau)} \ge 1-\prod_{\tau=t-\plength}^t (1-\PR{L_v(\tau)})\\
    &\ge 1-\exp\paren{-\sum_{\tau=t-\plength}^t L_v(\tau)}
  \end{align*}
  
  Finally since for $\tau \in [t-\plength, t]$ we have $L_v(\tau) \ge
  0.058$, then for a sufficiently large $c$ (i.e. $c \ge 18$) with high
  probability node $v$ is not competing at round $t$.
\end{proof}

Next, we show that if a node hears a beep or produces a beep at a round
when where its neighborhood (and its neighbors neighborhood) has a
``low'' beep potential, then with constant probability either it joins
the MIS, or one of its neighbors joins the MIS.

\begin{lem}
  Assume that $E_u(t) \le \half$ for every $u \in N(v)\cup\set{v}$.

  If node $v$ beeps or hears a beep at round $t$ then with probability at
  least $\frac{1}{e}$ either $v$ beeped alone, or one if its neighbors
  beeped alone.
  \label{lem:goodbeep}
\end{lem}

\begin{proof}
  We consider three events.
  \begin{align*}
    A_u&: \text{ Node $u$ beeps at round $t$.} \\
    B_u&: \text{ Node $u$ beeps alone at round $t$.} \\
    S &: \bigcup_{w \in N(v)\cup\set{v}} B_w
  \end{align*}
  Our aim is to show that the event $S$ happens with constant
  probability, as a first step we show that $\PR{B_u|A_u}$ is constant.

  \begin{align*}
  \PR{B_u|A_u} & =\PR{\neg{\bigcup_{w \in N(u)} A_w}} =
  \PR{\bigcap_{w \in N(u)} \neg{A_w}} = \prod_{w \in N(u)}(1-b_w(t))
  \\
  & \ge \exp\paren{-2\sum_{w \in N(u)} b_w(t)} = e^{-2 E_u(t)}
  \end{align*}

  Moreover, since by assumption $E_u(t)\le \half$ then $\PR{B_u|A_u}\ge
  \frac{1}{e}$.

  For simplicity we rename the set $N(v)\cup \set{v}$ to the set
  $\set{1,\ldots,k}$ where $k=|N(v)|+1$.
  We define the following finite partition of the probability space:

  \begin{align*}
    \xi_1 &= A_1\\
    \xi_2 &= A_2 \cap \neg{A_1} \\
    \xi_3 &= A_3 \cap \neg{A_2} \cap \neg{A_1}\\
    \vdots\\
    \xi_k &= A_k \cap \bigcap_{i=1}^{k-1} \neg{A_i}
  \end{align*}

  Recall that
  by assumption our probability space is conditioned on
  the event that ``node $v$ beeps or hears a beep at round $t$'', or in
  other words $\exists i \in [k]$ such that $A_i$ has occurred. Moreover,
  observe that $\bigcup_{i=1}^k \xi_i = \bigcup_{i=1}^k A_i$, and thus
  $\PR{\bigcup_{i=1}^k \xi_i}=1$.

  Since the events $\xi_1,\ldots,\xi_k$ are pairwise disjoint, by the
  law of total probability we have

  \begin{align*}
    \PR{S} = \sum_{i=1}^k \PR{S|\xi_i}\PR{\xi_i}.
  \end{align*}

  Finally since $\PR{S|\xi_i}=\PR{B_i|\xi_i} \ge \PR{B_i|A_i}
  \ge \frac{1}{e}$ we have $\PR{S} \ge \frac{1}{e} \sum_{i=1}^k
  \PR{\xi_i} = \frac{1}{e}$.
\end{proof}

Now we have the key ingredients necessary to prove that our algorithm
terminates.

\begin{lemma}
  With high probability, after $\bigO(\log^2 N\log n)$ consecutive
  rounds without a neighbor crashing, a node is either in the MIS or has
  a neighbor in the MIS.
\end{lemma}

\begin{proof}
  We say a node has an event at round $t$, if it beeps or hears a beep
  at round $t$. First we claim  that with high probability a node has an
  event every $\bigO(\log^2 N)$ rounds.
  Consider a node who does not hear a beep within $\bigO(\log^2 N)$ rounds
  (if it does hear a beep, the claim clearly holds). Then after
  $\bigO(\log^2 N)$ rounds it will reach line 7 and beep (with
  probability 1) and the claim follows.

  From \cref{lem:lowEu} we know that when a node decides to beep, with
  high probability the beep potential of its neighborhood is less than
  $\half$. We can use a union bound to say that when a node hears a
  beep, with high probability the beep was produced by a node with a
  beep potential less than $\half$.
  Therefore, we can apply \cref{lem:goodbeep} to say that with constant
  probability every time a node has an event, either the node joins the
  MIS (if it was not in the MIS already) or it becomes covered by an MIS
  node.

  Therefore, with high probability after $\bigO(\log n)$ events is
  either part of the MIS or it becomes covered by an MIS node. Since
  with high probability there is an event every $\bigO(\log^2 N)$
  rounds, this implies that with high probability a node is either
  inside the MIS or has a neighbor in the MIS after $\bigO(\log^2 N \log
  n)$ rounds.
\end{proof}

This completes the proof for Theorem \ref{thm:simpleMIS} and also
implies fast convergence.

\end{document}